\documentclass[aps,prl,twocolumn,superscriptaddress, amsmath, amssymb, amsfonts]{revtex4-1}

\usepackage{hyperref}
\usepackage{tikz}
\usepackage{amsthm}
\usepackage{comment}
\newcommand{\channelobj}{\ensuremath{N}}
\newcommand{\channel}{\ensuremath{N}}
\newcommand{\channelset}{\ensuremath{N}}

\newcommand{\chvec}{\ensuremath{n}}
\newcommand{\channelcu}[1]{\ensuremath{#1}}
\newcommand{\channelprob}[2]{\channelset(#2|#1)}
\newcommand{\channelcuprob}[3]{\channelcu{#1}(#3|#2)}

\newcommand{\inputs}{\ensuremath{\mathcal{X}}}
\newcommand{\outputs}{\ensuremath{\mathcal{Y}}}

\newcommand{\assistset}{\ensuremath{D}}
\newcommand{\assistobj}{\ensuremath{D}}
\newcommand{\assist}{\assistobj}

\newcommand{\assistcu}[1]{\ensuremath{#1}}
\newcommand{\suc}{\ensuremath{\mathrm{Succ}}}

\newcommand{\boxprob}[2]{\assistset(#2 | #1)}
\newcommand{\boxprobcu}[3]{#1(#3 | #2)}
\newcommand{\altchannel}{\ensuremath{T}}
\newcommand{\altchannelM}{\ensuremath{M}}

\newcommand{\aspace}{\ensuremath{V_A}}
\newcommand{\bspace}{\ensuremath{V_B}}
\newcommand{\apovm}{\ensuremath{A}}
\newcommand{\bpovm}{\ensuremath{B}}

\newcommand{\quantclass}{\textnormal{Q}}
\newcommand{\quantclassn}[1]{\textnormal{Q}(#1)}
\newcommand{\quantclassbin}{\textnormal{Q}_b}
\newcommand{\nsclass}{\textnormal{NS}}

\newcommand{\loc}{\textnormal{loc}}

\newcommand{\mkw}[1]{\textbf{[#1 --mary]}}

\newcommand{\radius}{\textnormal{Rad}}
\newcommand{\diameter}{\textnormal{Diam}}
\newcommand{\Tr}{\textnormal{Tr}}

\def\({\left(}
\def\){\right)}

\newcommand{\commentout}[1]{}

\usepackage{amsthm}
\newtheorem{theorem}{Theorem}
\newtheorem{proposition}[theorem]{Proposition}
\newtheorem{lemma}[theorem]{Lemma}
\newtheorem{corollary}[theorem]{Corollary}
\newtheorem{definition}[theorem]{Definition}

\begin{document}

\title{On optimal entanglement assisted one-shot classical communication}

\author{Brett Hemenway}
\email{bhemen@umich.edu}
\affiliation{Mathematics Department, University of Michigan, Ann Arbor, MI  48109, USA}

\author{Carl A.~Miller}
\email{carlmi@umich.edu}
\affiliation{Dept.~of Electrical Engineering and Computer Science, University of Michigan, Ann Arbor, MI  48109, USA}

\author{Yaoyun Shi}
\email{shiyy@umich.edu}
\affiliation{Dept.~of Electrical Engineering and Computer Science, University of Michigan, Ann Arbor, MI  48109, USA}

\author{Mary Wootters}
\email{wootters@umich.edu}
\affiliation{Mathematics Department, University of Michigan, Ann Arbor, MI  48109, USA}

\date{\today}

\begin{abstract}

\noindent The \em one-shot success probability \em of a noisy classical channel for transmitting one classical 
bit is the optimal probability with which the bit can be sent via a single use of the channel.
Prevedel {\em et al.} ({\em PRL} {\bf 106}, 110505 (2011)) recently showed 
that for a specific channel, this quantity can be increased if the parties using the channel
share an entangled quantum state. We completely characterize the optimal entanglement-assisted
protocols in terms of the radius of a set of operators 
associated with the channel. This characterization can be used to construct optimal entanglement-assisted
protocols from the given classical channel and to prove the limit of such protocols. As an example, we show that the Prevedel {\em et al.}
protocol is optimal for two-qubit entanglement.
We also prove some simple upper bounds on the improvement that can 
be obtained from quantum and no-signaling correlations.
\end{abstract}

\pacs{03.67.Hk, 03.65.Ud}

\maketitle

Suppose that two parties, Alice and Bob, communicate over a noisy classical channel.
While there are many examples of how Alice and Bob may benefit when they upgrade to a quantum channel, 
examples in which  shared entanglement improves communication
over a classical channel have only recently been discovered~\cite{clmw:2010,prevedeletal,clmw:2011}.
That these examples exist at all is somewhat surprising, as neither shared entanglement~\cite{bsst:2002} nor the assistance of non-signaling correlations~\cite{clmw:2011} can increase the classical capacity of the channel.
So far, work in this direction has focused on the
the \em (one-shot) zero error capacity\em, which measures the number of messages Alice can send to Bob perfectly~\cite{clmw:2011,mat:2012,bei:2010,lmmor:2012}, and the related notion of the \em one-shot success probability\em~\cite{prevedeletal}, which is the best probability with which Alice can successfully send a single bit to Bob.

It is of interest to determine how shared entanglement affects these two quantities, 
as this will further our understanding of how resources from quantum mechanics 
can be used for communication.

\commentout{ 
In this paper, we consider the one-shot success probability.  
We give a concise formula for the entanglement-assisted one-shot success probability, and we also give more easily computable bounds for certain types of entanglement. 
}

Previous work on enhanced communication over a classical channel has focused on the assistance that can be provided by non-signaling correlations.
In this setting, both the zero error capacity and one-shot success probability can be written as the solution to linear programs~\cite{clmw:2011, mat:2012}.  
Some upper bounds are known for the entanglement assisted zero error capacity~\cite{bei:2010}; these bounds are often the best bounds available in the unassisted case, suggesting that there are strong limitations to the amount of assistance that entanglement can provide.
Much less is known about the limits of quantum assistance for the one-shot success probability.
In~\cite{prevedeletal}, Prevedel \em et al\em.  give an example of a channel where the unassisted success probability,
$\suc(\channelobj),$ the entanglement-assisted success probability $\suc_{\quantclass}(\channelobj)$, and the non-signaling assisted success probability $\suc_{\nsclass}(\channelobj)$
are all different.  
It is known that entanglement cannot be completely helpful: if $\suc(\channelobj)$ is less than one, then so is $\suc_{\quantclass}(\channelobj)$~\cite{clmw:2011}.  However, the size of the gap between them has remained unquantified.

We use two distinct approaches to quantify the extent to which entanglement can help Alice and Bob.  In our first approach, we derive a simple formula for $\suc_{\quantclass}(\channelobj)$ in terms of the dimension of the entanglement.  This formula, which is given by maximizing a quantity over a family of positive semidefinite operators, is easy to work with, and as an example of its applicability, we show that the protocol from~\cite{prevedeletal} is in fact optimal for their channel and for 2-dimensional entanglement assistance.  

While our first approach is quite general, it does not give a closed form for the success probability. Our second approach
obtains explicit closed-form upper bounds for the success probability.  
As a first step, we prove the following general bound on non-signaling assistance.
Let $r$ be the number of elements in the input alphabet of $\channel$.  Then,
\begin{equation}\label{eq:nsbound}
\frac{\suc_{\nsclass} ( \channelobj ) - \frac{1}{2}}{\suc ( \channelobj ) - \frac{1}{2} } \leq 2 - \frac{2}{r}.
\end{equation}
The quantity $(\suc(\channelobj) - \frac{1}{2} )$ measures the advantage that Alice and Bob have over a random strategy; thus,
\eqref{eq:nsbound} measures the additional advantage gained by non-signaling correlations.
Our proof of \eqref{eq:nsbound} uses the linear program characterization of $\suc_{\nsclass} ( \channelobj )$ from
\cite{mat:2012}.  
From this, we derive an upper bound on the amount of assistance from a binary quantum device; we use the fact that
any quantum correlation can be decomposed
into a probabilistic mixture of a local correlation and a non-signaling correlation (the concept of \textit{local fraction}).  
We show that both of these bounds are the best possible, in the sense that there are channels for which equality is achieved.

A common thread in both approaches above is the use
of the \textit{radius} of a subset of a normed vector space.
Our formula for $\suc_{\quantclass} ( \channelobj )$ depends on maximizing
the radius of a family of Hermitian operators.  In the second approach we
use a formula for $\suc_{\nsclass} ( \channelobj )$
(an alternate formulation of Proposition 14 from \cite{mat:2012}) which is expressed
in terms of the radius of
a particular set of vectors.  

\paragraph*{Notation and terminology.}

Throughout this paper, we assume that Alice is trying to transmit
a single bit to Bob across a classical channel.
Alice and Bob will have access to a two-part input output device $\assistobj$ (Figure \ref{fig:nsdevice}),
which may be classical, quantum, or implement an arbitrary non-signaling correlation.
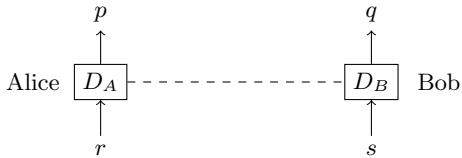
\begin{figure}[h]
\begin{center}
\begin{tikzpicture}[scale=0.9]
\node (alice) at (-1,0) {Alice};
\node (p) at (0, 1) {$p$};
\node[draw, rectangle] (da) at (0, 0) {$\assistobj_A$};
\node (r) at (0, -1) {$r$};
\draw (r) edge[->] (da);
\draw (da) edge[->] (p);
\node (bob) at (5,0) {Bob};
\node (q) at (4, 1) {$q$};
\node[draw, rectangle] (db) at (4, 0) {$\assistobj_B$};
\node (s) at (4, -1) {$s$};
\draw (s) edge[->] (db);
\draw (db) edge[->] (q);
\draw (da) edge[dashed] (db);
\end{tikzpicture}
\end{center}
\vskip-0.1in
\caption{A two-part input output device.} 
\label{fig:nsdevice}
\end{figure}
Each two-part input output device $\assist$ gives rise to a \em correlation \em between Alice and Bob, given by 
\[\left\{ \boxprob{rs}{pq} \mid p \in \mathcal{P} , q \in \mathcal{Q},
r \in \mathcal{R} , s \in \mathcal{S} \right\}
\]
so that $\boxprob{rs}{pq}$ is the probability of outputs $p$ and $q$ given inputs $r$ and $s$.
We will abuse notation by identifying the device $\assist$ with the correlation it induces.

We say a device is \em non-signaling \em if the partial sums
$\sum_{p \in \mathcal{P}} \boxprob{rs}{pq}$
do not depend on $r$, and the partial sums
$\sum_{q \in \mathcal{Q}} \boxprob{rs}{pq}$
do not depend on $s$.
We say a non-signaling device $\assist$ is
\textit{quantum} if there exist Hibert spaces $\aspace$
and $\bspace$, families of POVMs $\{ \{ A_r^p \}_p \}_r$
$\{ \{ B_s^q \}_q \}_s$, and a density operator $\Lambda$
on $\aspace \otimes \bspace$ such that
$\boxprob{pq}{rs} = \Tr ( ( A_r^p \otimes B_s^q ) \Lambda ).$
The device is quantum with \em dimension $n$ \em if both $\aspace$ and
$\bspace$ are $n$-dimensional, and \em binary \em if the input and output
alphabets have size $2$.

A classical channel $\channelobj$ is given by a matrix of conditional probabilities
$\{ \channelprob{x}{y} \mid x \in \mathcal{X},
y \in \mathcal{Y} \}$, where $\channelprob{x}{y}$ is the probability of seeing an output $y \in \outputs$ given the input $x \in \inputs.$
%
For any channel $\channelobj$, let $\suc ( \channelobj )$ denote the maximum probability
with which a single bit can be sent across $\channelobj$ (without assistance).
Let $\suc ( \channelobj, \assist )$
denote the maximum probability for a single-bit transmission across
$\channelobj$ with the assistance of $\assist$.
If $\mathbf{S}$ is a set 
of two-part devices,
write
$\suc_\mathbf{S} ( \channelobj ) := \sup_{S \in \mathbf{S}} \suc
(\channelobj, S ).$
We will be concerned with three choices of $\mathbf{S}$.  
We consider the set $\nsclass$ of non-signaling devices; the sets $\quantclass$ and $\quantclassn{n}$ of quantum and $n$-dimensional quantum devices; and the set $\quantclassbin$ of binary quantum devices.

\paragraph*{General quantum devices.}

In this section, we derive a formula for $\suc_{\quantclassn{n}}(\channel)$, and give an example of how to use our formula.
We will use the \em radius \em of a finite set $\{ H_i\}_{i \in \mathcal{I}}$ of Hermitian operators
on a finite-dimensional Hilbert space $V$, defined by
$
	\radius \{ H_i \}_i := \min_C \max_i \left\| H_i - C \right\|
$
where the minimum is taken over all Hermitian operators $C$ on $V$.  
The following lemma, which is proved in section 1 of the supplementary
material, gives an alternative expression for
the radius.
\begin{lemma}
\label{radiuslemma}
For any finite set $\{ H_i \}_{i \in \mathcal{I}}$ of Hermitian operators
on a finite-dimensional Hilbert space $V$, the radius of $\{ H_i \}$ is
equal to
\[ 
\max_{\substack{\lambda_i \geq 0, \lambda'_i \geq 0 \\ \sum \lambda_i
= \sum \lambda'_i \\ \Tr ( \sum \lambda_i ) = 1/2}} \left[ \sum_{i \in \mathcal{I}}
\Tr \left( ( \lambda_i - \lambda'_i ) H_i  \right) \right].
\]
\end{lemma}

Here, the maximization is over all Hermitian operators $\{ \lambda_i \}_{i \in
\mathcal{I} }$ and $\{ \lambda'_i \}_{i \in \mathcal{I}}$  on $V$ satisfying
the given constraints.

Using this lemma, we will prove the following theorem, which characterizes $\suc_{\quantclassn{n}}(\channel)$.

\begin{theorem}
\label{radiustheorem}
For any channel $\channel$, and any integer $n \geq 2$,
\begin{eqnarray*}
\suc_{\quantclassn{n}} ( \channel ) & = & \frac{1}{2} +
\max_{\{ B_y \}} \left( 
\radius \left\{ \sum_{y \in \outputs} 
\channelprob{x}{y} B_y \right\}_{x \in \inputs} \right),
\end{eqnarray*}
where the maximization is over all families
$\{ B_y \}_{y \in \outputs}$ of Hermitian operators
on $\mathbb{C}^n$ satisfying $0 \leq B_y \leq \mathbb{I}$.
\end{theorem}

\begin{proof}
Consider the following quantum-assisted protocol for transmitting
a single bit across $\channel$.
Alice and Bob possess a bipartite quantum system
represented by a density matrix $\Lambda$
on a Hilbert space $\aspace \otimes \bspace$.  
Alice wishes to transmit a message $a \in \{ 0, 1 \}$.
Depending the value of $a$, she applies
one of two possible POVMs $\{ \apovm_0^x \}_{x \in \inputs}$
or $\{ \apovm_1^x \}_{x \in \inputs}$ to $\aspace$ and sends the result
of the measurement to the channel $\channel$.  Bob receives
the output $y$ of the channel, and according to this output,
applies one of a family of binary POVMs
$\left\{ \{ \bpovm_y^0 , \bpovm_y^1 \} \right\}_{y \in \mathcal{Y}}$
to $\bspace$.  The result of this output is Bob's guess
at Alice's original message.

In order to compute the success probability for
this protocol, it is not necessary to know the
state $\Lambda$ or the operators $\{ A_a^x \}_{x,a}$:
it is only necessary to know the operators 
$\rho_a^x := \Tr_A \left[ ( \apovm_a^x \otimes \mathbb{I} )
\Lambda \right],$
which represent the state of Bob's quantum system
when the outcome of Alice's measurement is $x$.
These operators satisfy $\sum_x \rho_0^x = \sum_x \rho_1^x$
and $\Tr ( \rho_0^x ) = 1$, and, in fact, any family of operators
satisfying those two constraints can be induced by an appropriately
chosen state $\Lambda$ and appropriately chosen
POVMs $\{ \apovm_0^x \}_{x \in \inputs}$
or $\{ \apovm_1^x \}_{x \in \inputs}$.  Thus, for our purposes, to
specify an $(n,n)$-dimensional entanglement-assisted strategy for communicating
a single bit across $\channel$, it suffices to specify 
a collection of binary POVMs 
$\left\{ \left\{ B^0_y, B^1_y\right\}\right\}_{y \in \outputs}$ on $\mathbb{C}^n$ and
a collection of positive semidefinite operators
$\{ \rho_a^x \mid a \in \{ 0, 1 \}, x \in \inputs \}$ on $\mathbb{C}^n$
satisfying
\begin{equation}\label{eq:rhoconstraint}
\sum_x \rho_0^x  =  \sum_x \rho_1^x \qquad \text{and}\qquad
\Tr ( \sum_x \rho_0^x ) = 1.
\end{equation}
The success probability of the protocol
is given by
\begin{align*}
&\frac{1}{2}  \left( \sum_{ y \in \outputs} \channelprob{x}{y} \sum_{x \in \inputs}
\Tr  ( \rho_0^x \bpovm_y^0 )  \right) 
 + \\
& \qquad \frac{1}{2} \left( \sum_{y \in \outputs} \channelprob{x}{y}
\sum_{x \in \inputs} \Tr  ( \rho_1^x \bpovm_y^1 )  \right) .
\end{align*}
Let $B_y = B_y^0$.  Since $B_y^1 = \mathbb{I} - B_y$, the expression above simplifies to
\begin{equation}
\label{keysuccexp}
\frac{1}{2} + \frac{1}{2} \cdot \Tr \left[  \sum_{x \in \inputs}
\left( \rho_0^x - \rho_1^x \right) \sum_{y \in \outputs}
\channelprob{x}{y} B_y \right].
\end{equation}
The quantity $\suc_{\quantclassn{n} } ( \channel )$ is
the maximum of this expression over all 
$n \times n$ Hermitian operators $\{ B_y \}_{y \in \outputs}$ satisfying
$0 \leq B_y \leq \mathbb{I}$ and all $n \times n$ positive semidefinite operators
$\{ \rho_a^x \}_{x \in \inputs, a \in \{ 0 , 1 \}}$ satisfying 
\eqref{eq:rhoconstraint} above.
Applying Lemma~\ref{radiuslemma} yields the desired formula.
\end{proof}

A convexity argument (see section 2 of the supplementary information)
proves the following stronger version of Theorem~\ref{radiustheorem}.

\begin{corollary}
\label{radiuscorollary}
The formula in Theorem~\ref{radiustheorem} holds also
when the maximum is taken only over 
families $\{ B_y \}$ that consist of projections 
on $\mathbb{C}^n$.
\end{corollary}

As an example of the utility of Corollary \ref{radiuscorollary}, consider
the channel $\altchannelM$ in Figure~\ref{fig:prevchannel}, which is
defined in \cite{prevedeletal}.
The input alphabet
for $\altchannelM$ is $\{ 1 , 2, 3, 4 \}$, and the output alphabet
is $\{ 1, 2, 3, 4, 5, 6 \}$.  In section 3 of the supplementary information,
we prove that for any $2 \times 2$ projection operators
$P_1, \ldots , P_6$, 
the radius of the set
$\left\{ \sum_{y=1}^6 \channelcuprob{M}{x}{y} P_y \mid x = 1, 2, 3, 4 \right\}$
is no more than $\frac{1}{6} + \frac{1}{3 \sqrt{2}}$.  This maximum
is achieved when $P_1 = 0$, $P_2 = \mathbb{I}$, and 
$\{ P_3, P_4 \}$ and $\{ P_5, P_6 \}$ are two different Pauli measurements.
Therefore,
\[\suc_{\quantclassn{2}} \left( \altchannelM \right) =  
\frac{2}{3} + \frac{1}{3 \sqrt{2}},
\]
and the protocol from \cite{prevedeletal} is optimal for $2$-dimensional
entanglement assistance.
(We note that this generalizes the paper \cite{wb:2011}, which showed the optimality
of \cite{prevedeletal} within a more restricted class of protocols.)
\begin{figure}[h]
\begin{eqnarray*}
\begin{tabular}{r|c|c|c|c|c|c|}
& 1 & 2 & 3 & 4 & 5 & 6 \\
\hline
1 & 1/3 & 0   & 1/3 & 0 & 1/3 & 0  \\
\hline
2 & 1/3 & 0   & 0 & 1/3 & 0 & 1/3 \\
\hline
3 & 0    & 1/3& 1/3 & 0 & 0  & 1/3 \\
\hline
4 & 0    & 1/3& 0    &1/3&1/3 & 0 \\
\hline
\end{tabular}
\end{eqnarray*}
\caption{The channel $\altchannelM$, from~\cite{prevedeletal}.}
\label{fig:prevchannel}
\end{figure}

\paragraph*{Non-signaling devices.}
In order to prove more explicit bounds on the limits of quantum assistance, we first turn our 
attention to assistance by a non-signaling correlation.
The next proposition asserts a formula for the optimal
non-signaling assisted success probability of a channel.
For any finite set of vectors $S \subseteq \mathbb{R}^k$,
let $\radius_1 ( S )$ denote the radius of $S$ under
the $1$-norm.
\begin{proposition}
\label{nsassistprop}
Let $\channelobj$ be a classical channel,
and for each $x \in \inputs$, let $
\chvec_x 
= \{ \channelset ( y \mid x ) \}_{y \in \outputs}
\in \mathbb{R}^\outputs. $
Then,
\begin{eqnarray}
\label{nsassistpropformula}
\suc_{\nsclass} ( \channelobj ) & = & \frac{1}{2} +
\frac{1}{2} \cdot \radius_1 \left\{
\chvec_x \mid x \in \inputs \right\}.
\end{eqnarray}
\end{proposition}

Note that in the above formula, we take
the radius of $\{ \chvec_x \}$ \textit{as a subset
of $\mathbb{R}^\outputs$}, not
as a subset of the set of probability distributions
on $\outputs$.

\begin{proof}[Proof of Proposition~\ref{nsassistprop}]
By Proposition~14 from \cite{mat:2012},
\begin{eqnarray*}
\suc_{\nsclass} ( \channelobj ) & = & 
1 - \max_{c \in \mathbb{R}^\outputs}
\min_{x \in \inputs} \sum_{y \in \outputs}
\left( \min \{ c_y , \channelset ( y \mid x ) \}
- c_y/2 \right) \\
& = & \min_{c \in \mathbb{R}^\outputs}
\max_{x \in \inputs} \sum_{y \in \outputs}
\left( 1 - \min \{ c_y , \channelset ( y \mid x ) \} + 
 c_y/2 \right).
\end{eqnarray*}
Using the easily proved fact
that $\left\| u - v \right\|_1
= \sum_i u_i + \sum_i v_i - 2 \sum_i \min \{ u_i, v_i \}$,
the above formula simplifies to
\begin{eqnarray*}
\suc_{\nsclass} ( \channelobj ) 
& = & \min_{c \in \mathbb{R}^\outputs}
\max_{x \in \inputs} \left( \frac{1}{2}
+ \frac{1}{2} \cdot \left\| c - \chvec_x \right\|_1 \right),
\end{eqnarray*}
which implies (\ref{nsassistpropformula}) above.
\end{proof}

Formula (\ref{nsassistpropformula}) allows us to relate
the quantity $\suc_\nsclass ( \channelobj )$
to the quantity $\suc ( \channelobj )$.

\begin{theorem}
\label{thm:mainresultinit}
Let $\channelobj$ be a classical channel,
and let $r = \left| \inputs \right|$ denote the
size of the input alphabet of $\channelobj$.  Then,
\begin{eqnarray}
\label{eq:initbound}
\suc_\nsclass ( \channelobj  ) - \frac{1}{2} & \leq &
\left( 2 - \frac{2}{r} \right) \left[ \suc ( \channelobj ) - \frac{1}{2} \right].
\end{eqnarray}
\end{theorem}

\begin{proof}
Let $\{ \chvec_x \}$ be the vectors defined in
Proposition~\ref{nsassistprop}.
The unassisted one-shot success probability
can be expressed in terms of these vectors like so:
\begin{eqnarray}
\suc ( \channelobj ) & = & \frac{1}{2} +
\frac{1}{4} \cdot \diameter_1 \{ \chvec_x \},
\end{eqnarray}
where $\diameter_1$ denotes diameter
under the $1$-norm.  A triangle-inequality
argument shows that the distance from the mean vector
$\left( \sum_x \chvec_x \right)/r$ to the set $\{ \chvec_x \}$
cannot exceed $\left( 1 - \frac{1}{r} \right) \diameter_1
\{ \chvec_x \}$. Therefore,  $\radius_1 \{ \chvec_x \} \leq \left( 1 - \frac{1}{r} \right) \diameter_1 \{ \chvec_x \}$,
which implies the desired result.
\end{proof}

Theorem \ref{thm:mainresultinit} is the best possible in the sense that there are channels where equality is
achieved in \eqref{eq:initbound}.  Consider the following example, which is a generalization of the 
channel $\altchannelM$ from Figure \ref{fig:prevchannel}.
Let $s$ be a positive integer.  For any $i \in \{ 0, 1, 2, \ldots, 2^s - 1 \}$,
let $b_i \in \mathbb{F}_2^s$ denote the binary representation of $i$.
Define a channel $\altchannel$ as follows.  The input alphabet
of $\altchannel$ is $\{ 0, 1, 2, \ldots, 2^s - 1 \}$, 
and the output alphabet is $\{ 1, 2, \ldots, 2^s - 1 \} \times \{ 0 , 1 \}$.
On given input $i$, the channel chooses an element $j \in 
\{ 1, \ldots, 2^s - 1 \}$ uniformly at random and outputs
the pair $(j , b_i \cdot b_j )$ (where $b_i \cdot b_j$ denotes
the inner product of $b_i$ and $b_j$ mod 2).

For any $i \in \{ 0,  1, 2, \ldots, 2^s - 1 \}$, let
$\ell_i \in \mathbb{R}^{2 ( 2^s - 1 )}$ denote the probability
vector which expresses the output of $\altchannel$
on input $i$.  It is easy to see that the diameter of
$\{ \ell_i \}$ is $2^s / (2^s - 1 )$, and thus 
$\suc ( \altchannel ) = \frac{1}{2} + 2^{s-2}/(2^s - 1 )$.
On the other hand, the radius of $\{ \ell_i \}$ is $1$,
as can be seen from the following calculation.  For
any $c \in \mathbb{R}^{2 (2^s - 1 )}$,
\begin{align*}
&\max_{0 \leq i \leq 2^s - 1}
\left\| \ell_i - c \right\|_1
\geq 2^{-s} \sum_{i = 0}^{2^s - 1}
\left\| \ell_i - c \right\|_1 \\
&\  = 2^{-s} \sum_{\substack{1 \leq j \leq 2^s - 1 \\
t \in \{ 0, 1 \} }} \left[  2^{s-1} \left| c_{jt} - (2^s - 1)^{-1} \right|
+ 2^{s-1} \left| c_{jt} - 0\right| \right] \\
&\ \geq  2^{-s}  \sum_{\substack{1 \leq j \leq 2^s - 1 \\
t \in \{ 0, 1 \} }} \left[ 2^{s-1} (2^s - 1)^{-1} \right] = 1.
\end{align*}
Thus $\suc_{\nsclass} ( \altchannel ) = 1$.  (And, indeed,
a perfect communication protocol for $\altchannel$ exists--see
section 4 of the supplementary information.) 
The channel $\altchannel$ achieves equality in (\ref{eq:initbound}).

The following modifiied version of Theorem \ref{thm:mainresultinit} will be
useful in our analysis of entanglement assistance.

\begin{theorem}
\label{thm:mainresult}
Let $\channel$ be a classical channel, and
let $\assist$ be a non-signaling correlation arising from a two-part
device $( \assist_A , \assist_B )$.  Let $m$ denote the size of the output
alphabet of $\assist_A$.  Then,
\begin{eqnarray}
\suc ( \channelobj , \assist  ) - \frac{1}{2} & \leq &
\left( 2 - \frac{1}{m} \right) \left[ \suc ( \channelobj ) - \frac{1}{2} \right].
\end{eqnarray}
\end{theorem}

\begin{proof}
A protocol for communicating a single bit $a$ using $\channelobj$ and
$\assist$ proceeds as follows. Alice uses $a$ to choose an input to $\assist_A$, and
then uses $a$ and the output of $\assist_A$ to choose an input to $\channelobj$.
Bob uses the output of $\channelobj$ to choose an input to $\assist_B$, 
and then uses the outputs of $\channelobj$ and $\assist_B$ together to guess
the bit $a$.

The optimal success probability $\suc ( \channelobj , \assist )$ can be
achieved by a \textit{deterministic} protocol (i.e., a protocol in
which Alice and Bob make their choices according to deterministic functions).
As there are only $2m$ possible inputs that Alice could make
to $\channelobj$ in a deterministic protocol, the success probability
of such a protocol is bounded by $(2 - 2/(2m) ) \suc (\channelobj)$
by Theorem~\ref{thm:mainresultinit}.
\end{proof}

\paragraph*{Binary quantum devices.}
Finally, we will use our bounds for non-signaling devices
to obtain bounds for assistance by binary quantum devices.  

A two-part device $\assist$ is \textit{local-deterministic}
if the output of each part is a deterministic function of its input.
A non-signaling correlation
is \textit{local} if it is a convex combination of local-deterministic
correlations.
We define the \textit{local fraction} of a
non-signaling correlation, a concept which is used in \cite{elitzuretal},
\cite{barrettetal}. 

\begin{definition}
Let $\assistset$ be a non-signaling correlation.  The local fraction
of $\assistset$, denoted $\loc (\assistset)$, is the largest real number $\alpha \in [ 0, 1 ]$ such that there exists
a decomposition
\begin{equation}
\label{locdecomp}
\assist = \alpha L + (1 - \alpha) F,
\end{equation}
where $L$ is a local correlation and $F$ is a non-signaling correlation.
\end{definition}

For any classical channel $\channelobj$, it is easy to see that when a decomposition such as \eqref{locdecomp}
exists with $L$ local and $F$ non-signaling,
\begin{eqnarray*}
\suc ( \channel , \assist ) & \leq & \alpha \suc ( \channel , L ) + ( 1 - \alpha ) \suc ( \channel , F ) \\
& \leq & \alpha \suc ( \channel )  + ( 1 - \alpha ) \suc_{\nsclass} ( \channel ).
\end{eqnarray*}
This implies the following stronger version of Theorem~\ref{thm:mainresult}.
\begin{theorem}
\label{thm:mainresult2}
Let $\channel$ be a channel, and let $\assist$ be a non-signaling correlation arising
from a two-part device $(\assistobj_A , \assistobj_B )$.
Let $m$ denote the size of the output alphabet of $\assistobj_A$.  Then
\[
	\frac{\suc ( \channel, \assist ) - \frac{1}{2}}{\suc ( \channel ) - \frac{1}{2}} \leq  1 +  \left( 1 -  \frac{1}{m} \right)(1 - \loc(\assist)). \qed
\]
\end{theorem}
Thus, to obtain improved upper bounds on $\suc ( \channel, \assist )$ for quantum correlations
$\assist$, it suffices to find lower bounds on the local fractions
of quantum correlations.  
In section 5 of the supplementary material, we use facts about the geometry of quantum and non-signaling correlations~\cite{tsirelson} to prove the following bound for binary quantum correlations.
\begin{proposition}
\label{binquantprop}
Let $\assist$ be a binary 
quantum correlation.  Then
$\loc ( \assist ) \geq 2 - \sqrt{2}.$ 
\end{proposition}

Combining Theorem~\ref{thm:mainresult2} and Proposition~\ref{binquantprop} yields
the following.

\begin{corollary}
\label{bquantcor}
For any classical channel $\channel$, 
\[
\frac{ \suc_{\quantclassbin} ( \channel )  - \frac{1}{2} }{ \suc ( \channel ) - \frac{1}{2} }
\leq \frac{1}{2} + \frac{1}{\sqrt{2}}. 
\]
\end{corollary}

Note that equality occurs in Corollary~\ref{bquantcor} for
the case discussed in \cite{prevedeletal}.

\paragraph*{Conclusion.}
We have given a formula for the $n$-dimensional entanglement-assisted one-shot success probability
of a classical channel, and have shown its utility by using it to show that the protocol in~\cite{prevedeletal} is optimal.
We derived a more explicit bound on the advantage gained by binary quantum correlations (which is an equality in the case of \cite{prevedeletal}).  Along the
way, we established a bound on the advantage gained by non-signaling assistance and provided an example where equality is achieved.

Future research could explore methods for evaluating the formula from Theorem~\ref{radiustheorem}.
(Section 3 of the supplementary information provides methods which might generalize.)
Also, it would be interesting to try to prove stronger bounds on the increase in
$\suc ( \channel )$ that is provided by entanglement.  (This might involve generalizations of Proposition~\ref{binquantprop}.)
Another natural next step would be to consider the one-shot success probability for non-binary messages.

The authors would like to thank Vincent Russo for his help with the preparation and editing
of this paper.  We also thank Shmuel Friedland, Aubrey da Cunha, Xiaodi Wu, and Kim Winick for many 
useful discussions, and the anonymous PRL reviewers for helpful comments.  This research was supported in part by
the National Basic Research Program of China under Awards 2011CBA00300 and 2011CBA00301,
and the NSF of the United States under Award 1017335.

\bibliographystyle{apsrev}
\bibliography{NS_assist}

\end{document}


\title{Supplementary Information}

\maketitle

\section{The radius of a set of Hermitian operators}

\begin{lemma}
For any finite set $\{ H_i \}_{i \in \mathcal{I}}$ of Hermitian operators
on a finite-dimensional Hilbert space $V$, the radius of $\{ H_i \}$ is
equal to
\begin{eqnarray*}
 \max_{\substack{\lambda_i \geq 0, \lambda'_i \geq 0 \\ \sum \lambda_i
= \sum \lambda'_i \\ \Tr ( \sum \lambda_i ) = 1/2}} \left[ \sum_{i \in \mathcal{I}}
\Tr \left( ( \lambda_i - \lambda'_i ) H_i  \right) \right].
\end{eqnarray*}
\end{lemma}

\begin{proof}
Any family of Hermitian operators $\{H_i\}$ may be translated
to a family $\{ H_i + W \}$ which contains the operator $0$.
This translation does not affect the radius nor the expression from the
statement of the lemma.  Therefore, we may assume that
$\{ H_i \}$ contains $0$.
By definition,
\begin{eqnarray}
\radius \{ H_i \}_i & = & \min_{ \substack{C, r \\ C - H_i \geq - r \mathbb{I} \\
H_i - C \geq - r \mathbb{I}}}  ( r ),
\end{eqnarray}
where the maximization is over Hermitian operators $C$ and real numbers $r$.
Since $0 \in \{ H_i \}$, whenever the constraints
in this maximization are satisfied we have in particular
that $C \geq - r \mathbb{I}$.
Letting $Z = C + r \mathbb{I}$, we obtain
the following alternate expression:
\begin{eqnarray}
\radius \{ H_i \}_i & = & \min_{ \substack{Z, r \\ Z \geq H_i  \\ -Z
+ 2 r \mathbb{I} \geq - H_i \\
}}  ( r ).
\end{eqnarray}
By semidefinite programming duality, this is equivalent to
\begin{eqnarray*}
\radius \{ H_i \}_i & = & \max_{ \substack{\lambda_i \geq 0, \lambda'_i \geq 0, \\ \sum_i \lambda_i
- \sum_i \lambda'_i \leq 0 \\ 2 \Tr ( \sum \lambda'_i ) \leq 1}} \left[ \left( \sum \Tr ( \lambda_i H_i )
- \sum \Tr ( \lambda'_i H_i ) \right) \right].
\end{eqnarray*}
It is easy to see that this maximum is achieved by a
pair of families $\{ \lambda_i \}, \{ \lambda'_i \}$ satisfying $\sum \lambda_i = \sum \lambda'_i$ and
$2 \Tr ( \sum \lambda'_i ) = 1$.
\end{proof}

\section{The proof of Corollary 3 in the main text}

\label{cor3thm2sec}

The radius function is convex in the following sense:
for any familes of operators $\{ J_y \}_{y \in \outputs}$
and $\{ K_y \}_{y \in \outputs}$, and real number
$\alpha \in [ 0, 1 ]$, 
\begin{align*}
&\radius \{ \alpha J_y + (1-\alpha ) K_y \}_y \\
&\qquad \leq \alpha \radius \{ J_y \}_y + (1 - \alpha ) 
\radius \{ K_y \}_y.
\end{align*}
(For, if we let $J'$ be such that the distance
from $J'$ to $\{ J_y \}_y$ is equal to $r := \radius \{ J_y \}_y$,
and we let $K'$ be such that the distance
from $K'$ to $\{ K_y \}_y$ is equal to $r' := \radius \{ K_y \}_y$,
then the distance from $\alpha J' + (1 - \alpha ) K'$
to $\{ \alpha J_y + ( 1 - \alpha ) K_y \}_y$ is no more
than $\alpha r + (1 - \alpha ) r'$ by the triangle inequality.)
In particular, this convexity property implies that 
the radius of $\{ \alpha J_y + ( 1 - \alpha ) K_y \}_y$ is no 
more than the maximum of $\radius \{ J_y \}_y$ and $\radius \{ K_y \}_y$.

Since any Hermitian operator $B$ satisfying $0 \leq B \leq \mathbb{I}$ is
a convex combination of projection operators, 
Corollary 3 follows from Theorem 2.

\section{An example calculation}

Let $\altchannelM$ be the channel defined in figure 2 in 
the main text.  In this section we will use
Theorem 2 from the main text to calculate the quantity 
$\suc_{\quantclassn{2}} ( \altchannelM )$.  First, we will 
prove the following lemma which provides a simplified formula 
for $\suc_{\quantclassn{n}} ( \altchannelM )$.
For any projection operator $P$, let $P^\perp$ denote projection
onto the orthogonal complement of $P$.
\begin{lemma}
\label{simplifiedradiuslemma}
For any $n \geq 1$, the quantity $\suc_{\quantclassn{n}} ( \altchannelM )$
is equal to
\begin{eqnarray*}
& & \frac{1}{2} + \left( \frac{1}{3} \right)
\max_{X, Y, Z } \left( \radius 
\left\{ X + Y +Z  , X + Y^\perp + Z^\perp ,
X^\perp + Y + Z^\perp , X^\perp + Y^\perp + Z \right\} \right),
\end{eqnarray*}
where the maximum is taken over all projection operators $X, Y, Z$ on
$\mathbb{C}^n$.
\end{lemma}

\begin{proof}
For any Hermitian operators $B_1, B_2, B_3, B_4, B_5, B_6$ on $\mathbb{C}^n$,
let
\begin{eqnarray*}
F ( B_1, B_2, B_3, B_4, B_5, B_6 )
\end{eqnarray*}
be equal to the quantity
\begin{eqnarray*}
& & \radius \left\{ B_1 + B_3 + B_5, 
B_1 + B_4 + B_6, B_2 + B_3 + B_6 , B_2 + B_4 + B_5 \right\}.
\end{eqnarray*}
By the formula from Theorem 2 in the main text,
\begin{eqnarray}
\suc_{\quantclassn{n}} ( \altchannelM ) & = &
\frac{1}{2} + \left( \frac{1}{3} \right) \max_{0 \leq B_i \leq \mathbb{I} } 
F ( B_1, B_2, B_3, B_4, B_5, B_6 ).
\end{eqnarray}
Let
\begin{eqnarray}
m & = & \max_{0 \leq B_i \leq \mathbb{I} } F ( B_1, B_2, B_3, B_4, B_5, B_6 ).
\end{eqnarray}
It suffices to prove that this maximum is achieved by some
$6$-tuple of the form $(X, X^\perp , Y ,  Y^\perp , Z, Z^\perp )$, 
where $X$, $Y$, and $Z$ are projections.

As noted in section~\ref{cor3thm2sec} of the supplementary information,
the radius function is convex in the sense that
if $(H_1, H_2, H_3, H_4 )$ and $(H'_1, H'_2, H'_3 , H'_4 )$ are Hermitian
operators and $\alpha \in [ 0, 1 ]$ is a real number,
\begin{eqnarray}
\radius \{ \alpha H_i + (1 - \alpha ) H'_i \}_i & \leq &
\alpha \radius \{  H_i \}_i +
(1 - \alpha ) \radius \{  H'_i \}_i.
\end{eqnarray}
It follows easily by linearity that a similar convexity property
holds for $F$: for any Hermitian operators $B_1, \ldots , B_6 $
and $B'_1 , \ldots, B'_6 $, and any $\alpha \in [ 0, 1 ]$,
\begin{eqnarray*}
& &F ( \alpha B_1 + (1 - \alpha ) B'_1 , \ldots, \alpha B_6 + (1 - \alpha ) B'_6 ) \\
\nonumber & \leq &  \alpha F ( B_1, \ldots, B_6 ) + (1 - \alpha ) F ( B'_1, \ldots, B'_6 ).
\end{eqnarray*}
In particular,
\begin{eqnarray}
\label{convexityimp}
& &F ( \alpha B_1 + (1 - \alpha ) B'_1 , \ldots, \alpha B_6 + (1 - \alpha ) B'_6 ) \\
\nonumber & \leq &  \max \left\{  F ( B_1, \ldots, B_6 ), F ( B'_1, \ldots, B'_6 ) \right\}.
\end{eqnarray}
Additionally, $F$ is translation-invariant in the following sense:
for any Hermitian operators $B_1, \ldots, B_6$,
and any Hermitian operators $K$, $L$, and $M$,
\begin{eqnarray}
\label{translationinvariance}
F ( B_1 +  K , B_2 + K , B_3 + L , B_4 + L, B_5 + M,
B_6 + M) & = & F ( B_1, \ldots, B_6 ).
\end{eqnarray}

Let $X_1, X_2, Y_1, Y_2, Z_1, Z_2$ be Hermitian operators satisfying $0 \leq X_i, Y_i, Z_i \leq \mathbb{I}$
such that $F ( X_1, X_2, Y_1, Y_2, Z_1, Z_2 ) = m$.  Let $X_+$ and $X_-$ be a pair of 
positive semidefinite operators having mutual orthogonal supports which are such that
\begin{eqnarray}
X_1 - X_2 = X_+ - X_-.
\end{eqnarray}
Define $Y_+, Y_-, Z_+, Z_-$ similarly.  
By property (\ref{translationinvariance}) above,
\begin{eqnarray}
F ( X_+ , X_-, Y_+ , Y_-, Z_+, Z_- ) & = & F ( X_1, X_2, Y_1, Y_2, Z_1 , Z_2 ) \\
\nonumber & = & m.
\end{eqnarray}

The pair $(X_+, X_-)$ can be expressed as a convex combination of pairs of
projections $(P_1^{(i)} , P_2^{(i)} )$ where for each $i$, the support of $P_1^{(i)}$
is orthogonal to $P_2^{(i)}$.  A similar decomposition exists for $(Y_+, Y_-)$ and
$(Z_+, Z_- )$.  Therefore by property (\ref{convexityimp}) above, there
exist pairs of projections $(P_1, P_2)$, $(Q_1, Q_2 )$, $(R_1, R_2 )$, with each pair
having mutually orthogonal supports, such that
\begin{eqnarray}
F ( P_1, P_2, Q_1, Q_2, R_1, R_2 ) = m.
\end{eqnarray}

Let $P_3 = \mathbb{I} - P_1 - P_2$, and define $Q_3$ and $R_3$ similarly.
By (\ref{translationinvariance}),
\begin{eqnarray}
F \left(  P_1 + \frac{P_3}{2} , P_2 + \frac{P_3}{2}, Q_1 + \frac{Q_3}{2}
 , Q_2 + \frac{Q_3}{2} , R_1 + \frac{R_3}{2} , R_2 + \frac{R_3 }{2} \right) & = & m.
\end{eqnarray}
The $6$-tuple on the left hand side of the equation above is
a convex combination of the $6$-tuples
\begin{eqnarray*}
& & \left( P_1 + P_3 , P_2, Q_1 + Q_3 , Q_2 , R_1 + R_3, R_2 \right) \\
& & \textnormal{ and } \left( P_1, P_2 + P_3, Q_1 , Q_2 + Q_3 , R_1, R_2 + R_3 \right).
\end{eqnarray*}
By (\ref{convexityimp}), at least one of these $6$-tuples must achieve the maximum $m$.
This completes the proof.
\end{proof}

\begin{lemma}
For any projection operators $X$, $Y$, $Z$ on the two-dimensional
vector space $\mathbb{C}^2$,
the radius of the set
\begin{eqnarray}
\label{radset}
\left\{ X + Y +Z  , X + Y^\perp + Z^\perp ,
X^\perp + Y + Z^\perp , X^\perp + Y^\perp + Z \right\}.
\end{eqnarray}
is less than or equal to $\frac{1}{2} + \frac{1}{\sqrt{2}}$.
\end{lemma}

\begin{proof}
\textbf{Case 1: The matrices $X$, $Y$, and $Z$ are all scalar matrices.}
In this case, each of $X$, $Y$, and $Z$ is equal to either $0$ or $\mathbb{I}$.
This case is trivial, since the radius of the set $\{ 3 \mathbb{I} , \mathbb{I} \}$
is $1$, and the radius of the set $\{ 2 \mathbb{I}, 0 \}$ is $1$.

\vskip0.1in

\textbf{Case 2: Two of the matrices $X, Y, Z$ are scalar matrices and one
is a nonscalar.}  We may assume without loss of generality that $X$ is
the nonscalar matrix.  Then the set (\ref{radset}) is equal to either
\begin{eqnarray}
\left\{ 0 ,  X + \mathbb{I} , 2 \mathbb{I} \right\}
\end{eqnarray}
or
\begin{eqnarray}
\left\{ X , X+ 2 \mathbb{I} , \mathbb{I} \right\}.
\end{eqnarray}
In the former case, the operator-norm
distance from the operator $\mathbb{I}$
to the set $\{ 0, X + \mathbb{I} , 2 \mathbb{I} \}$ is $1$.
In the latter case, the operator-norm distance
from the operator $X + \mathbb{I}$ to the set
$\{ X, X + 2 \mathbb{I} , \mathbb{I} \}$ is $1$.  The desired result follows.

\vskip0.1in

\textbf{Case 3: Exactly one of the matrices $X, Y, Z$ is a scalar matrix.}
We may assume that $X$ and $Y$ are nonscalar matrices and $Z$ is scalar.
Also, by replacing $(X, Y, Z )$ with $(X^\perp, Y , Z^\perp)$ if necessary, we
may assume that $Z = \mathbb{I}$.  

Let $X = \left| x \right> \left< x \right|$ and $Y = \left| y \right> \left< y \right|$
where $x, y \in \mathbb{C}^2$ are unit vectors, and let $\theta = \arccos \left( \left| x \cdot y
\right| \right)$.  Both of the operators
\begin{eqnarray}
X + Y + \mathbb{I} , X^\perp + Y^\perp + \mathbb{I}
\end{eqnarray}
have eigenvalues $\left\{ 2 + \cos \theta, 2 - \cos \theta \right\}$, and
both of the operators
\begin{eqnarray}
X + Y^\perp, X^\perp +  Y
\end{eqnarray}
have eigenvalues $\left\{ 1 + \sin \theta , 1 - \sin \theta \right\}$.  If we let
\begin{eqnarray}
C = \left( \frac{3}{2} + \frac{\cos \theta - \sin \theta }{2} \right) \mathbb{I},
\end{eqnarray}
then the operator norm distance from $C$ to each of the elements
 of (\ref{radset}) is $\frac{1}{2} + \frac{\cos \theta + \sin \theta}{2} \leq
\frac{1}{2}  + \frac{1}{\sqrt{2}}$.

\vskip0.1in

\textbf{Case 4: Each of $X, Y, Z$ is a nonscalar matrix.}  As in case 3,
let $X = \left| x \right> \left< x \right|$ and $Y = \left| y \right> \left< y \right|$
and let $\theta = \arccos \left( \left| x \cdot y
\right| \right)$.

Let
\begin{eqnarray}
C = \mathbb{I} + \left( \frac{1}{2} + \frac{\cos \theta - \sin \theta}{2} \right) 
Z + \left( \frac{1}{2} - \frac{\cos \theta - \sin \theta}{2} \right) Z^\perp.
\end{eqnarray}
Then, the operator norm of the difference
\begin{eqnarray}
( X + Y + Z ) - C = (X + Y) - \left( \frac{3}{2} + 
\frac{ \cos \theta - \sin \theta }{2} \right) \mathbb{I} 
\end{eqnarray}
is $\frac{1}{2} + \frac{\cos \theta + \sin \theta}{2}$, which is
less than or equal to $\frac{1}{2}  + \frac{1}{\sqrt{2}}$.  A similar
calculation shows that the distance from $C$ to each of the
other three elements of set $(\ref{radset})$ is equal
to $\frac{1}{2} + \frac{\cos \theta + \sin \theta}{2}$.  This
completes the proof.
\end{proof}

For any angle $\theta \in \mathbb{R}$,
let $P_\theta \colon \mathbb{C}^2 \to \mathbb{C}^2$
denote projection onto the unit vector $\cos (\theta)  \left| 0 \right>
+ \sin ( \theta ) \left| 1 \right>$.  Consider the set
\begin{eqnarray}
\label{maxradiusset1}
\left\{ P_0 + P_{\pi / 4 } + \mathbb{I} , P_0 + P_{3 \pi / 4 } ,
P_{\pi / 2} + P_{\pi / 4} , P_{\pi / 2 } + P_{3 \pi / 4 } + \mathbb{I} \right\}
\end{eqnarray}
A direct calculation  shows that
the distance from the operator $\left( \frac{3}{2} \right) \mathbb{I}$ to
set (\ref{maxradiusset1}) is $\frac{1}{2} + \frac{1}{\sqrt{2}}$.  The next lemma
asserts that this quantity is in fact the radius of (\ref{maxradiusset1}).

\begin{lemma}
\label{directcalclemma}
The radius of the set 
\begin{eqnarray}
\left\{ P_0 + P_{\pi / 4 } + \mathbb{I} , P_0 + P_{3 \pi / 4 } ,
P_{\pi / 2} + P_{\pi / 4} , P_{\pi / 2 } + P_{3 \pi / 4 } + \mathbb{I} \right\}
\end{eqnarray}
is 
is equal to $\frac{1}{2} + \frac{1}{\sqrt{2}}$.
\end{lemma}

\begin{proof}
For any Hermitian operator $H \colon \mathbb{C}^2 \to \mathbb{C}^2$,
let us write $\overline{H}$ to denote the trace-zero operator $H - (\Tr (H ) ) \mathbb{I} / 2$.
In the proof that follows, we will make use of the following fact:
for any two Hermitian operators $Q, R \colon \mathbb{C}^2 \to \mathbb{C}^2$,
\begin{eqnarray}
\left| \left| Q - R \right| \right| =
\left| \Tr ( Q ) - \Tr ( R ) \right| + \left| \left| \overline{Q} - \overline{R} \right| \right|
\end{eqnarray}

Suppose, for the sake of contradiction, that there exists a Hermitian
operator $Z$ whose distance from each of the elements of set (\ref{maxradiusset1})
is strictly less than $\frac{1}{2} + \frac{1}{\sqrt{2}}$.  Then,
\begin{eqnarray*}
2 \left( \frac{1}{2} + \frac{1}{\sqrt{2}} \right) & > & 
\left\|  ( P_0 + P_{3 \pi / 4 } ) - Z  \right\|  + \left\|  ( P_{\pi / 2 }
+ P_{\pi / 4 } ) - Z \right\| \\
& = & \left\| \left( P_0 + P_{3 \pi / 4 } - \mathbb{I} \right) - \overline{Z} \right\|
+ \left\| \left( P_{\pi /2 } + P_{\pi / 4 } - \mathbb{I} \right) - \overline{Z} \right\|
+ 2 \cdot \left| 2 - \Tr ( Z ) \right|\\
 & \geq &
\left\| \left( P_0 + P_{3 \pi / 4 } \right)
- \left( P_{\pi / 2 } + P_{\pi / 4 } \right) \right\| + 
2 \cdot \left| 2 - \Tr ( Z ) \right| \\
 & = &
\sqrt{2} + 
2 \cdot \left| 2 - \Tr ( Z ) \right| 
\end{eqnarray*}
Therefore, $\Tr (Z ) < \frac{5}{2}$.  Similarly,
\begin{eqnarray*}
2 \left( \frac{1}{2} + \frac{1}{\sqrt{2}} \right) & > & 
\left\|  ( P_0 + P_{ \pi / 4 } + \mathbb{I} ) - Z  \right\|  + \left\|  ( P_{\pi / 2 }
+ P_{3 \pi / 4 } ) - \mathbb{I} \right\| \\
& = & \left\| \left( P_0 + P_{ \pi / 4 } - \mathbb{I} \right) - \overline{Z} \right\|
+ \left\| \left( P_{\pi /2 } + P_{3 \pi / 4 } - \mathbb{I} \right) - \overline{Z} \right\|
+ 2 \cdot \left| 3 - \Tr ( Z ) \right|\\
 & \geq &
\left\| \left( P_0 + P_{\pi / 4 } \right)
- \left( P_{\pi / 2 } + P_{3 \pi / 4 } \right) \right\| + 
2 \cdot \left| 3 - \Tr ( Z ) \right| \\
 & = &
\sqrt{2} + 
2 \cdot \left| 3 - \Tr ( Z ) \right|,
\end{eqnarray*}
which implies $\Tr ( Z ) > \frac{5}{2}$.  This is a contradiction.
\end{proof}

Combining Lemmas \ref{simplifiedradiuslemma}--\ref{directcalclemma}, we
have the following proposition.
\begin{proposition}
The quantity $\suc_{\quantclassn{2}} ( \altchannelM )$ is
equal to $\frac{2}{3} + \frac{1}{3 \sqrt{2}}$.  $\qed$
\end{proposition}

\section{An example of optimal non-signaling assistance}

In this section we discuss an example in which equality occurs in Theorem 5 from the main text.  This
example is a generalization of the protocol from \cite{prevedeletal}.

Let $m$ be a positive integer. Let 
\begin{eqnarray}
\altinputs & = & \mathbb{F}_2^m, \\
\altoutputs & = & \left( \mathbb{F}_2^m \smallsetminus \{ 0 \} \right) \times \mathbb{F}_2.
\end{eqnarray}
Let $\altchannelK$ be a channel defined as follows:
\begin{enumerate}
\setlength{\parskip}{0pt}
\setlength{\itemsep}{0pt}

\item The input alphabet of $\altchannelK$ is 
$\altinputs$, and the output alphabet of $\altchannelK$
is $\altoutputs$.

\item For any given input $\mathbf{v} \in \mathbb{F}_2^m$,
the output of $\altchannelK$ is unformly distributed over the
set
\begin{eqnarray}
\left\{ \left( \mathbf{w}, \mathbf{w} \cdot \mathbf{v} \right) \mid \mathbf{w} \in \mathbb{F}_2^m \smallsetminus
\{ 0 \} \right\}.
\end{eqnarray}
\end{enumerate}
(Here, $\mathbf{w} \cdot \mathbf{v} \in \mathbb{F}_2$ denotes the inner product of
$\mathbf{w}$ and $\mathbf{v}$.)  

Let $(\altassistbox_1 , \altassistbox_2 )$ be a two part
input-output device defined as follows.  (See Figure~\ref{fig:exampledevice}.)
\begin{enumerate}
\setlength{\parskip}{0pt}
\setlength{\itemsep}{0pt}
\item The input alphabet for $\altassistbox_1$
is $\mathbb{F}_2$, and the output alphabet for
$\altassistbox_1$ is $\altinputs$.

\item The input alphabet for $\altassistbox_2$ 
is $\altoutputs$, and the output alphabet
for $\altassistbox_2$ is $\mathbb{F}_2$.

\item If the inputs to $\altassistbox_1$
and $\altassistbox_2$ are $a \in \{ 0, 1 \}$
and
$
(\mathbf{w} , r ) \in \left( \mathbb{F}_2^m \smallsetminus
\{ \mathbf{0} \} \right) \times \mathbb{F}_2,
$
then the output of $\altassistbox_1$ is unformly distributed
over all vectors $\mathbf{a} = \left( a_1, a_2, \ldots , a_m \right)$
that satisfy $a_1 = a$, and the output of $\altassistbox_2$ is
$a \oplus r \oplus \left( \mathbf{w} \cdot \mathbf{a} \right)$.
\end{enumerate}

\begin{figure}
\begin{center}
\begin{tikzpicture}
\node (input1) at (0, 2) {$a$};
\node (input2) at (4, 2) {$\mathbf{w}, r$};
\node[draw, rectangle] (e1) at (0, 1) {$\altassistbox_1$};
\node[draw, rectangle] (e2) at (4, 1) {$\altassistbox_2$};
\node (output1) at (0, 0) {$\mathbf{a} =  ( a, a_2, \ldots , a_m)$};
\node (output2) at (4, 0) {$a \oplus r \oplus \left( \mathbf{w} \cdot \mathbf{a} \right)$};
\draw (input1) edge[->] (e1);
\draw (input2) edge[->] (e2);
\draw (e1) edge[->] (output1);
\draw (e2) edge[->] (output2);
\end{tikzpicture}
\end{center}
\caption{The device $(\altassistbox_1, \altassistbox_2 )$.}
\label{fig:exampledevice}
\end{figure}

It can be checked that the correlation $\altassist$ arising
from $(\altassistbox_1 , \altassistbox_2 )$ is non-signaling.
Additionally, one can see (by substitution) that using
$\altassist$ to assist $\altchannelK$ yields a perfect transmission
of a single bit.  (See figure~\ref{fig:perfectcommunication}.)

\begin{figure}
\begin{center}
\begin{tikzpicture}[scale=0.9]
\node (alice) at (-1,0) {Alice};
\node[draw] (d1) at (2,0) {$\altassistbox_1$};
\node[draw, rectangle] (channel) at (3,-1) {\altchannelK};
\node[draw, rectangle] (d2) at (4,-2) {$\altassistbox_2$};
\node (bob) at (7,-2) {Bob};
\draw (alice) edge[->] node[auto]{$a$} (d1);
\draw (d1) edge[->] node[auto]{$\mathbf{a}$} (channel);
\draw (channel) edge[->] node[auto]{$\left( \mathbf{w} , \mathbf{w} \cdot \mathbf{a} \right)$} (d2);
\draw (d2) edge[->] node[auto]{$a$} (bob);
\end{tikzpicture}
\end{center}
\caption{A perfect communication protocol.}
\label{fig:perfectcommunication}
\end{figure}
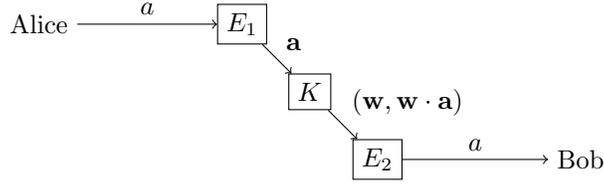

Now, let us calculate the quantity $\suc \left( \altchannelK \right)$.
For any two distinct vectors $\mathbf{x}_0 , \mathbf{x}_1 \in
\mathbb{F}_2^m$, the probability that a randomly chosen
vector $\mathbf{w} \in \mathbb{F}_2^m \smallsetminus \{ \mathbf{ 0 } \}$
will satisfy
$\mathbf{w} \cdot \mathbf{x}_0 \neq \mathbf{w} \cdot \mathbf{x_1}$
is equal to $2^{m-1} / \left( 2^m - 1 \right)$.
This fact has the following consequence: if Alice
employs the deterministic
encoding strategy $\left[ 0 \mapsto \mathbf{x}_0 , 1
\mapsto \mathbf{x}_1 \right]$ to send a single bit,
then the optimal probability with which Bob can decode
is
\begin{eqnarray}
& & \left[ \frac{2^{m-1}}{2^m - 1 } \right] (1)
+ \left[ \frac{2^{m-1} - 1}{2^m - 1 }  \right] \left( \frac{1}{2} \right) \\
\label{succexample}
& = & \frac{2^m + 2^{m-1} - 1}{2^{m+1} - 2}.
\end{eqnarray}
Therefore, $\suc \left( \altchannelK \right)$ is equal
to quantity (\ref{succexample}), while
$\suc_\nsclass \left( \altchannelK \right)$ is equal to $1$.
Theorem 5 from the main text asserts the following bound
on $\suc_\nsclass \left( \altchannelK \right)$:
\begin{eqnarray*}
\suc_\nsclass \left( \altchannelK \right) & \leq & \frac{1}{2} + \left( 2 - \frac{2}{2^m} \right) \left[
\suc ( \altchannelK ) - \frac{1}{2} \right] \\
& = & \frac{1}{2} + 2 \left( \frac{2^m - 1}{2^m } \right)
\left( \frac{2^{m-1} } {2^{m+1} - 2} \right) \\
& = & 1.
\end{eqnarray*}
Therefore, equality is achieved in Theorem 5 from the main text
when $\channel= \altchannelK$.

\section{The local fraction of a binary quantum correlation}

In this section, we prove the following proposition from the main text.

\begin{proposition}
Let $\assist$ be a binary 
quantum correlation.  Then
$\loc ( \assist ) \geq 2 - \sqrt{2}.$ 
\end{proposition}

\begin{proof}
For any binary non-signaling correlation $\assistcu{G}$, let
\[
f_1 \left( \assistcu{G} \right)
 =  \sum_{a, x, b, y \in \{ 0, 1 \} }
(-1)^{x \oplus b \oplus ( a \wedge y )} \boxprobcu{G}{ay}{xb}.
\]
This is the function which defines the CHSH inequality \cite{chsh}. 
let $f_2$, $f_3$, and $f_4$ be the functions defined by the same
expression with $a \wedge y$ replaced by $\neg a \wedge y$,
$a \wedge \neg y$, and $\neg a \wedge \neg y$, respectively.

We note the following facts.  (See \cite{tsirelson}.)
\begin{enumerate}
\setlength{\parskip}{0pt}
\setlength{\itemsep}{0pt}

\item \label{localcrit} A non-signaling correlation $\assistcu{G}$ is local 
if and only if $
-2 \leq f_i \left( \assistcu{G} \right) \leq 2
$
for $i = 1, 2, 3, 4$.

\item If $\assistcu{G}$ is a quantum correlation,
then for $i = 1,2,3,4$,
\begin{equation*}
\label{quantbound}
- 2 \sqrt{2} \leq f_i \left( \assistcu{G} \right) \leq 2 \sqrt{2}.
\end{equation*}

\item There are eight non-signaling
correlations $\{ \assistcu{P}_i^+ \}_{i=1}^4$
and $\{ \assistcu{P}_i^- \}_{i=1}^4$,
satisfying
\begin{eqnarray*}
f_j \left( \assistcu{P}_i^\pm \right) & = &
\begin{cases}  \pm 4 & \textnormal{ if }
j = i \\
0 & \textnormal{ otherwise} \end{cases}
\end{eqnarray*}
These are the \em Popescu-Rohrlich \em  (PR) boxes.

\item \label{gendecomp} Every non-signaling correlation
is a convex combination of local correlations
and the eight PR boxes.  Further, for any two distinct PR boxes $P$ and $P'$, the 
correlation $(P + P')/2$ is local.
\end{enumerate}

From the second part of item \ref{gendecomp}, it follows that any
convex combination of local boxes and PR boxes
can be simplified into an expression of the form
$\alpha \assistcu{L} + (1 - \alpha ) \assistcu{Q}$,
where $\assistcu{L}$ is local, $\assistcu{Q}$
is a PR box, and $\alpha \in [0, 1]$.
Any non-signaling correlation can thus
be expressed as a convex combination
of a local correlation and a single PR box.

Let $\assist = \alpha \assistcu{L} + (1 - \alpha) \assistcu{Q}$,
where $\assistcu{L}$ is local and $\assistcu{Q}$
is a PR box.
First suppose that $\assistcu{Q} = \assistcu{P}_j^+$.  
Let
$\assistcu{L}_\beta = (\alpha \assistcu{L}
+ (\beta - \alpha) \assistcu{P}_j^+)/\beta.$
for any $\beta \in [ \alpha , 1 ]$.
Then $\assistcu{L}_\beta$ is local whenever $f_j ( \assistcu{L}_\beta
) \leq 2$.
If $f_j \left( \assistcu{L}_1 \right) < 2$,
then $\assistcu{L}_1 (= \assist)$ is local, and the
proposition follows easily.  Otherwise, there is
a value $\beta \in [\alpha , 1]$ such that $f_j \left( \assistcu{L}_\beta
\right) = 2$.  We have
$\assist  =  \beta \cdot \assistcu{L}_\beta +
(1 - \beta ) \assistcu{P}_j^+.$
The quantity $\beta$ must be at least $2 - \sqrt{2}$,
since otherwise \eqref{quantbound} would be violated.
Therefore $\loc (D ) \geq 2 - \sqrt{2}$.

A similar argument completes the proof
in the case where $Q = P_j^-$.
\end{proof}

\bibliographystyle{plain}
\bibliography{NS_assist}